\documentclass[10pt, conference, twocolumn]{IEEEtran}
\usepackage{amssymb}
\usepackage{amsmath}
\usepackage{graphicx}
\usepackage{amsfonts}
\usepackage{dsfont}%
\usepackage{psfrag}
\IEEEoverridecommandlockouts
\newtheorem{thm}{Theorem}[section]

\newtheorem{prop}[thm]{Proposition}
\newtheorem{defn}{Definition}[section]
\newtheorem{rem}{Remark}[section]
\begin{document}

\title{Relaying Simultaneous Multicast Messages}
\author{D. G\"{u}nd\"{u}z$^{1,2}$, O. Simeone$^{3}$, A. Goldsmith$^{1}$, H. V.
Poor$^{2}$ and S. Shamai (Shitz)$^{4}$\\$^{1}$Dept. of Electrical Engineering, Stanford Univ., Stanford, CA 94305, USA\\$^{2}$Dept. of Electrical Engineering, Princeton Univ., Princeton, NJ 08544, USA\\$^{3}$CWCSPR, New Jersey Institute of Technology, Newark, NJ 07102, USA\\$^{4}$Dept. of Electrical Engineering, Technion, Haifa, 32000,
Israel \thanks{This work was supported by U.S. National Science Foundation
under grants CNS-06-26611 and CNS-06-25637, the DARPA ITMANET program under
grant 1105741-1-TFIND, the U.S. Army Research Office under MURI award
W911NF-05-1-0246 and by the Israel Science Foundation and the European Commission in the framework of the FP7 Network of Excellence in Wireless COMmunications NEWCOM++.}}

\maketitle

\begin{abstract}
The problem of multicasting multiple messages with the help of a relay, which
may also have an independent message of its own to multicast, is considered.
As a first step to address this general model, referred to as the compound
multiple access channel with a relay (cMACr), the capacity region of the
multiple access channel with a \textquotedblleft cognitive\textquotedblright%
\ relay is characterized, including the cases of partial and rate-limited
cognition. Achievable rate regions for the cMACr model are then presented
based on decode-and-forward (DF) and compress-and-forward (CF) relaying
strategies. Moreover, an outer bound is derived for the special case in which
each transmitter has a direct link to one of the receivers while the
connection to the other receiver is enabled only through the relay terminal.
Numerical results for the Gaussian channel are also provided.

\end{abstract}


\thispagestyle{empty}

\section{Introduction}

\label{s:intro}

Consider two non-cooperating satellites each multicasting radio/TV signals to
users on Earth. The coverage area and the quality of the transmission is
limited by the strength of the direct links from the satellites to
the users. To extend coverage, to increase capacity or to improve robustness,
a standard solution is that of introducing relay terminals, which may be other
satellite stations or stronger ground stations. The role of the relay
terminals is especially critical for users that lack a direct link from any of the satellites.

Cooperative transmission has been extensively studied both for a single user
with a dedicated relay terminal \cite{Cover:IT:79}, \cite{Kramer:IT:05} and
for two cooperating users \cite{Willems:IT:83}. In this work, we study cooperation in a model with two source terminals simultaneously multicasting independent information to two receivers
with the help of a relay. While the source terminals cannot directly cooperate, the relay can support both transmissions simultaneously to enlarge the multicast capacity region. Moreover, it is assumed that the relay station has also its own
message to be multicast.

The model under study is a \emph{compound multiple access channel with a
relay} (cMACr) and can be seen as an extension of several channel
models, for example, the compound multiple access channel (MAC), the broadcast channel
and the relay channel. The main goal of this work is to provide achievable rate regions and an outer bound on the capacity region for this model. We start our analysis by studying a simplified version of the cMACr that consists of a MAC with a ``cognitive'' relay. In this scenario the cognitive
relay is assumed to know both messages non-causally. We
provide the capacity region for this model. As an
intermediate step between the cognitive relay model and cMACr, we also consider the relay with finite capacity unidirectional links from the transmitters and find its capacity region. In this scenario, parts of the messages are transmitted to the relay over the finite capacity links and the rates of these links determine how much the relay can help each user. This is not the case in the general cMACr model since decoding at the relay might be restrictive, yet we can use these techniques to obtain achievable rate regions.

We provide achievable rate regions for cMACr with decode-and-forward (DF) and
compress-and-forward (CF) relaying. In the CF scheme, the relay, instead of
decoding the messages, quantizes and broadcasts its received signal. This
corresponds to the joint source-channel coding problem of broadcasting a
common source to two receivers, each with its own correlated side information,
in a lossy fashion, studied in \cite{Nayak:IT:08}. This result indicates that
the pure channel coding rate regions for certain multi-user networks can be
improved by exploiting related joint source-channel coding techniques. The cMACr model is also studied in \cite{Maric:MILCOM:07}, where DF and amplify-and-forward (AF) based protocols are analyzed, assuming that no private relay message is available.

\psfrag{W1}{$W_1$}\psfrag{W2}{$W_2$}\psfrag{W3}{$W_3$}
\psfrag{X1}{$X_1$}\psfrag{X2}{$X_2$}\psfrag{X3}{$X_3$}
\psfrag{Y1}{$Y_1$}\psfrag{Y2}{$Y_2$}\psfrag{Y3}{$Y_3$}
\psfrag{pxy}{${\textstyle p(y_1,y_2,y_3|x_1,x_2,x_3)}$}
\psfrag{hW1}{${\scriptstyle \hat{W}_1(1), \hat{W}_2(1), \hat{W}_3(1)}$}
\psfrag{hW2}{${\scriptstyle \hat{W}_1(2), \hat{W}_2(2), \hat{W}_3(2)}$}
\psfrag{Rel}{\small Relay}\psfrag{S1}{\small Encoder 1}\psfrag{S2}{\small Encoder 2}
\psfrag{D1}{\small Decoder 1}\psfrag{D2}{\small Decoder 2}
\psfrag{Ch}{cMACr}
\begin{figure}
\centering
\includegraphics[width=3.5in]{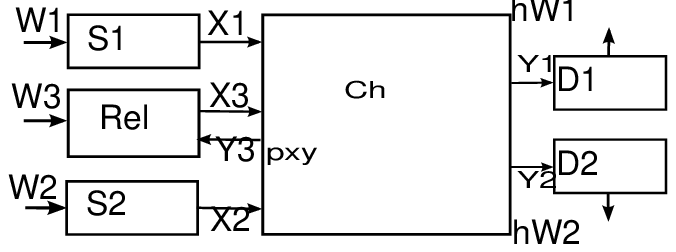}
\caption{A compound MAC with a relay (cMACr).}
\label{model}
\end{figure}

The rest of the paper is organized as follows. The system model is introduced
in Section \ref{s:system_model}. In Section \ref{s:MAC_cogrel} we study several cognitive relay models. The compound multiple access channel with a relay is studied in Section \ref{s:cMACr}. In Section \ref{p:GcMACr} numerical results for the Gaussian channel setup are presented. Section
\ref{s:conc} concludes the paper.

\section{System Model}\label{s:system_model}

A compound multiple access channel with relay consists of three input
alphabets $\mathcal{X}_{1}$, $\mathcal{X}_{2}$ and $\mathcal{X}_{3}$ of
transmitter 1, transmitter 2 and the relay, respectively, and three
output alphabets $\mathcal{Y}_{1}$, $\mathcal{Y}_{2}$ and $\mathcal{Y}_{3}$ of
receiver 1, receiver 2 and the relay, respectively. We consider a discrete
memoryless time-invariant channel without feedback characterized by $p(y_{1},y_{2},y_{3}|x_{1},x_{2},x_{3})$ (see Fig. \ref{model}). Transmitter $i$ has message $W_{i}\in\mathcal{W}_{i}$, $i=1,2$,
while the relay terminal also has a message $W_{3}\in\mathcal{W}_{3}$ of its
own, all of which need to be transmitted reliably to both receivers.

\begin{defn}
A $(2^{nR_{1}},2^{nR_{2}},2^{nR_{3}},n)$ code for the cMACr consists of three
sets $\mathcal{W}_{i}=\{1,\ldots,2^{nR_{i}}\}$ for $i=1,2,3$, two encoding
functions $f_{i}$ at the transmitters, $i=1,2$, $f_{i}:\mathcal{W}%
_{i}\rightarrow\mathcal{X}_{i}^{n}$, a set of (causal) encoding functions
$g_{j}$ at the relay, $j=1,\ldots,n$, $g_{j}:\mathcal{W}_{3}\times
\mathcal{Y}_{3}^{j-1}\rightarrow\mathcal{X}_{3}$, and two decoding functions
$h_{i}$ at the receivers, $i=1,2$, $h_{i}:\mathcal{Y}_{i}^{n}\rightarrow
\mathcal{W}_{1}\times\mathcal{W}_{2} \times\mathcal{W}_{3}$.
\end{defn}


We assume that the relay terminal is capable of full-duplex operation, i.e.,
it can receive and transmit at the same time instant. The average error probability is defined as
\[
P_{e}^{n}\triangleq\frac1{2^{n(R_{1}+R_{2}+R_{3})}} \sum_{\mathbf{W}} \Pr\left[  \bigcup_{i=1,2}\{\hat{\mathbf{W}}(i) \neq \mathbf{W} \}\right] ,
\]
where we defined $\mathbf{W}\triangleq(W_{1}, W_{2}, W_{3})$ and $\hat{\mathbf{W}}(i)
\triangleq(\hat{W}_{1}(i), \hat{W}_{2}(i), \hat{W}_{3}(i))$.

\begin{defn}
A rate triplet $(R_{1},R_{2},R_{3})$ is said to be \emph{achievable} for the
cMACr if there exists a sequence of $(2^{nR_{1}},2^{nR_{2}},2^{nR_{3}},n)$
codes with $P_{e}^{n}\rightarrow0$ as $n\rightarrow\infty$.
\end{defn}

\begin{defn}
The \emph{capacity region} $\mathcal{C}$ for the cMACr is the closure of the
set of all achievable rate triplets.
\end{defn}

\section{MAC with a Cognitive Relay}

\label{s:MAC_cogrel}

As stated in Section \ref{s:intro}, before addressing the general cMACr model we first study the MAC with a cognitive relay scenario in which the messages $W_{1}$ and
$W_{2}$ are assumed to be non-causally available at the relay terminal (in a
\textquotedblleft cognitive\textquotedblright\ fashion \cite{Devroye:IT:06})
and there is only one receiver ($\mathcal{Y}_{2}=\mathcal{Y}_{3}=\emptyset$
and $\mathcal{Y}=\mathcal{Y}_{1}$). The next proposition provides the capacity
region for this model.


\begin{prop}
\label{prop:1} For the MAC with a cognitive relay, the capacity region is the
closure of the set of all non-negative $(R_{1},R_{2},R_{3})$ satisfying
\begin{subequations}
\label{region cognitive}%
\begin{align}
R_{3}  &  \leq I(X_{3};Y|X_{1},X_{2},U_{1},U_{2},Q),\\
R_{1}+R_{3}  &  \leq I(X_{1},X_{3};Y|X_{2},U_{2},Q),\\
R_{2}+R_{3}  &  \leq I(X_{2},X_{3};Y|X_{1},U_{1},Q), \mbox{ and }\\
R_{1}+R_{2}+R_{3}  &  \leq I(X_{1},X_{2},X_{3};Y|Q) \label{non-conditional}%
\end{align}
for some joint distribution of the form
\end{subequations}
\begin{equation}
p(q)p(x_{1},u_{1}|q)p(x_{2},u_{2}|q)p(x_{3}|u_{1},u_{2},q)p(y|x_{1}%
,x_{2},x_{3}) \label{pmf}%
\end{equation}
for some auxiliary random variables $U_{1}$, $U_{2}$ and $Q$.
\end{prop}

\begin{proof}
The capacity region of a MAC with three users and any combination of ``common
messages'' (i.e., messages known ``cognitively'' to more than one user) is
given in \cite{Han}.
\end{proof}

We next consider the cases of partial and limited-rate cognition.

\begin{prop}
\label{prop:2} The capacity region of the MAC with a partially cognitive relay
(informed only of $W_{1})$ is the closure of the set of all non-negative
$(R_{1},R_{2},R_{3})$ satisfying
\begin{subequations}
\label{partially cognitive}%
\begin{align}
R_{2} &  \leq I(X_{2};Y|X_{1},X_{3},Q),\\
R_{3} &  \leq I(X_{3};Y|X_{1},X_{2},Q),\\
R_{1}+R_{3} &  \leq I(X_{1},X_{3};Y|X_{2},Q),\\
R_{2}+R_{3} &  \leq I(X_{2},X_{3};Y|X_{1},Q),\mbox{ and }\\
R_{1}+R_{2}+R_{3} &  \leq I(X_{1},X_{2},X_{3};Y|Q).
\end{align}
for an input distribution of the form $p(q)p(x_{2}|q)p(x_{1},x_{3}|q).$
\end{subequations}
\end{prop}


\begin{proof}
The proof, which we present in \cite{Gunduz:ITtbs}, is skipped here due to lack of space.
\end{proof}

\begin{rem}
The capacity region in Proposition \ref{prop:2} follows, like Proposition
\ref{prop:1}, from the capacity result by \cite{Han} for MAC with certain correlated sources. However, the
formulation given in (\ref{partially cognitive}) is more convenient than the
one obtained from \cite{Han} since, in the case of partial cognition, the
capacity region characterization does not require auxiliary random variables
in addition to the time-sharing random variable $Q$. This is because, unlike
the scenario covered by Proposition \ref{prop:1}, in which the relay's
codeword can depend on both $W_{1}$ and $W_{2}$, and the auxiliary random
variables quantify the amount of dependence on each message, for Proposition
\ref{prop:2}, the relay cooperates with only one source, and no auxiliary
random variable is needed.
\end{rem}

The MAC with a cognitive relay model can be further generalized to a scenario
with \textit{limited-capacity cognition}, in which the sources are connected
to the relay via finite-capacity orthogonal links, rather than having a priori
knowledge of the terminals' messages. In particular, assume that terminal $i$
can communicate with the relay, prior to transmission, via a link of
capacity\ $C_{i}$ for $i=1,2$. The following proposition establishes the
capacity region of this model.

\begin{prop}\label{prop:conf}
The capacity region of the MAC with a cognitive relay
connected to the source terminals via (unidirectional) links of capacities
$C_{1}$ and $C_{2}$ is given by
\begin{subequations}
\begin{align*}
R_{1}  &  \leq I(X_{1};Y|X_{2},X_{3},U_{1},U_{2},Q)+C_{1},\\
R_{2}  &  \leq I(X_{2};Y|X_{1},X_{3},U_{1},U_{2},Q)+C_{2},\\
R_{3}  &  \leq I(X_{3};Y|X_{1},X_{2},U_{1},U_{2},Q),\\
R_{1}+R_{2}  &  \leq I(X_{1},X_{2};Y|X_{3},U_{1},U_{2},Q)+C_{1}+C_{2},\\
R_{1}+R_{3}  &  \leq\min\left\{
\begin{array}
[c]{l}%
I(X_{1},X_{3};Y|X_{2},U_{1},U_{2},Q)+C_{1},\text{ }\\
I(X_{1},X_{3};Y|X_{2},U_{2},Q)
\end{array}
\right. \\
R_{2}+R_{3}  &  \leq\min\left\{
\begin{array}
[c]{l}%
I(X_{2},X_{3};Y|X_{1},U_{1},U_{2},Q)+C_{2}\\
I(X_{2},X_{3};Y|X_{1},U_{1},Q)
\end{array}
\right. \\
R_{1} + R_{2}+R_{3}  &  \leq\min\{ I(X_{1},X_{2},X_{3};Y|U_{1},Q)+C_{1},\\
&  \hspace{-.3in} I(X_{1},X_{2},X_{3};Y|U_{2},Q)+C_{2}, I(X_{1},X_{2}%
,X_{3};Y|Q),\\
&  \hspace{-.3in} I(X_{1},X_{2},X_{3};Y|U_{1},U_{2},Q)+C_{1}+C_{2} \}
\end{align*}
for auxiliary random variables $U_{1}, U_{2}$ and $Q$ with joint distribution
of the form (\ref{pmf}).
\end{subequations}
\end{prop}

\begin{proof}
The proof, which we present in \cite{Gunduz:ITtbs}, is skipped here due to lack of space.
\end{proof}

\begin{rem}
\label{r:cMAC_cogrel} Based on the results of this section, we can now take a
further step towards the analysis of the cMACr of Fig. \ref{model} by
considering the \textit{compound MAC with a cognitive relay}. This channel is shown in Fig. \ref{model} with the only difference that the relay here is
informed \textquotedblleft for free\textquotedblright\ of the messages $W_{1}$
and $W_{2}$ and that the signal received at the relay is non-informative,
e.g., $\mathcal{Y}_{2}=\emptyset.$ The capacity of such a channel follows
easily from Proposition \ref{prop:1} by taking the union over the distribution
$p(q)p(x_{1},u_{1}|q)p(x_{2},u_{2}|q)p(x_{3}|u_{1},u_{2},q)p(y_{1},y_{2}|x_{1},x_{2},x_{3})$ of the intersection of the two rate regions
(\ref{region cognitive}) evaluated for the two outputs $Y_{1}$ and $Y_{2}$.
Notice that this capacity region depends on the channel inputs only through
the marginal distributions $p(y_{1}|x_{1},x_{2},x_{3})$ and $p(y_{2}%
|x_{1},x_{2},x_{3}).$
\end{rem}


\section{Inner and Outer bounds on the Capacity Region of the Compound MAC
with a Relay}

\label{s:cMACr}

In this section, we focus on the general cMACr model illustrated in Fig.
\ref{model}. Single-letter characterization of the capacity region for this
model is open even for various special cases. Our goal here is to provide
inner and outer bounds.

The following inner bound is obtained by the DF scheme.
The relay fully decodes messages of both users so that we have a MAC from the
transmitters to the relay terminal. Once the relay has decoded the messages,
the transmission to the receivers takes place similarly to the MAC with a
cognitive relay model of Section \ref{s:MAC_cogrel}.

\begin{prop}\label{p:achDF}
For the cMACr as seen in Fig. \ref{model}, any rate triplet
$(R_{1},R_{2},R_{3})$ with $R_{j}\geq0$, $j=1,2,3$, satisfying
\begin{subequations}\label{ach:DF}%
\begin{align}
R_{1}  &  \leq I(X_{1};Y_{3}|U_{1}, X_{2},X_{3},Q),\label{ach 1}\\
R_{2}  &  \leq I(X_{2};Y_{3}|U_{2}, X_{1},X_{3},Q), \\
R_{1}+R_{2}  &  \leq I(X_{1},X_{2};Y_{3}|U_{1}, U_{2}, X_{3},Q),\label{ach 3}\\
R_{3}  &  \leq\min\{I(X_{3};Y_{1}|X_{1},X_{2},U_{1},U_{2},Q),\nonumber\\
&  ~~~~~ I(X_{3};Y_{2}|X_{1},X_{2},U_{1},U_{2},Q)\}, \label{ach 4}\\
R_{1}+R_{3}  &  \leq\min\{I(X_{1},X_{3};Y_{1}|X_{2},U_{2},Q), \nonumber\\
&  ~~~~~ I(X_{1},X_{3};Y_{2}|X_{2},U_{2},Q)\}, \\
R_{2}+R_{3}  &  \leq\min\{I(X_{2},X_{3};Y_{1}|X_{1},U_{1},Q), \nonumber\\
&  ~~~~~ I(X_{2},X_{3};Y_{2}|X_{1},U_{1},Q)\} \mbox{ and } \\
R_{1}+R_{2}+R_{3}  &  \leq\min\{I(X_{1},X_{2},X_{3};Y_{1}|Q), \nonumber\\
&  ~~~~~ I(X_{1},X_{2},X_{3};Y_{2}|Q)\} \label{ach 7}%
\end{align}
for auxiliary random variables $U_{1}, U_{2}$ and $Q$ with a joint distribution of the form $p(q)p(x_{1},u_{1}|q)p(x_{2},u_{2}|q)p(x_{3}|u_{1},u_{2},q)$ $p(y_{1},y_{2},y_{3}|x_{1},x_{2},x_{3})$ is achievable by DF.
\end{subequations}
\end{prop}

\begin{proof}
The proof follows by combining the block-Markov transmission strategy with DF
at the relay studied in \cite{Kramer:IT:05}, the joint encoding to handle the
private relay message and backward decoding at the receivers. Notice that
conditions (\ref{ach 1})-(\ref{ach 3}) ensure correct decoding at the relay,
whereas (\ref{ach 4})-(\ref{ach 7}) follow similarly to Proposition
\ref{prop:1} and Remark \ref{r:cMAC_cogrel} ensuring correct decoding at both receivers.
\end{proof}

Next, we consider applying the CF strategy at the relay terminal. In CF relaying introduced in \cite{Cover:IT:79}, the relay does not decode
the source message, but facilitates decoding at the destination by
transmitting a quantized version of its received signal. In quantizing its
received signal, the relay takes into consideration the correlated received
signal at the destination terminal and applies Wyner-Ziv source compression
(see \cite{Cover:IT:79} for details). In the cMACr scenario, unlike the
single-user relay channel, we have two distinct destinations, each with
different side information correlated with the relay received signal. This
situation is similar to the problem of lossy broadcasting of a common source
to two receivers with different side information sequences considered in
\cite{Nayak:IT:08} (and solved in some special cases), and applied to the
two-way relay channel setup in \cite{Gunduz:All:08}. Here, for simplicity, we
consider broadcasting only a single quantized version of the relay received
signal to both receivers. The following proposition states the corresponding
achievable rate region.

\begin{prop}
\label{p:achCF} For the cMACr of Fig. \ref{model}, any rate triplet
$(R_{1},R_{2},R_{3})$ with $R_{j}\geq0$, $j=1,2,3$, satisfying
\begin{align*}
R_{1}  &  \leq\min\{I(X_{1};Y_{1},\hat{Y}_{3}|X_{2},X_{3},Q), I(X_{1}%
;Y_{2},\hat{Y}_{3}|X_{2},X_{3},Q)\},\\
R_{2}  &  \leq\min\{I(X_{2};Y_{2},\hat{Y}_{3}|X_{1},X_{3},Q), I(X_{2}%
;Y_{1},\hat{Y}_{3}|X_{1},X_{3},Q)\},
\end{align*}
and \vspace{-.15in}
\begin{align*}
R_{1}+R_{2}\leq\min\{  &  I(X_{1},X_{2};Y_{1},\hat{Y}_{3}|X_{3},Q),\\
&  ~~~~~~ I(X_{1},X_{2};Y_{2},\hat{Y}_{3}|X_{3},Q)\}
\end{align*}
such that
\vspace{-.07in}
\begin{align*}
R_{3}+I(Y_{3};\hat{Y}_{3}|X_{3},Y_{1},Q)  &  \leq I(X_{3};Y_{1} |Q)
\mbox{ and }\\
R_{3}+I(Y_{3};\hat{Y}_{3}|X_{3},Y_{2},Q)  &  \leq I(X_{3};Y_{2} |Q)
\end{align*}
for random variables $\hat{Y}_{3}$ and $Q$ with a joint distribution
$p(q,x_{1},x_{2},x_{3},y_{1},y_{2},y_{3},\hat{y}_{3}) = p(q)p(x_{1}|q)
p(x_{2}|q) p(x_{3}|q)$ $p(\hat{y}_{3}|y_{3},x_{3},q) p(y_{1},y_{2},y_{3}%
|x_{1},x_{2},x_{3})$ is achievable with $\hat{Y}_{3}$ having bounded cardinality.
\end{prop}

\begin{proof}
The proof, which we present in \cite{Gunduz:ITtbs}, is skipped here due to lack of space.
\end{proof}

\begin{rem}
The achievable rate region given in Proposition \ref{p:achCF} can
potentially be improved. Instead of broadcasting a single quantized version of
its received signal, the relay can transmit two descriptions so that the
receiver with an overall better quality in terms of its channel from the relay
and the side information received from its transmitter, receives a better
description, and hence higher rates (see \cite{Nayak:IT:08} and
\cite{Gunduz:All:08}). Another possible extension which we will
not pursue here is to use the partial DF scheme together with CF \cite{Cover:IT:79}, \cite{Gunduz:All:08}.
\end{rem}

We are now interested in studying the special case in which each source
terminal can reach only one of the destination terminals directly. Assume, for
example, that there is no direct connection between source terminal 1 and
destination terminal 2, and similarly between source terminal 2 and
destination terminal 1. In practice, this setup might model either a larger
distance between the disconnected terminals, or some physical constraint in
between the terminals blocking the connection. In such a scenario, the relay
is essential in providing coverage to multicast data to both receivers. We
model this scenario by the (symbol-by-symbol) Markov chain conditions:
\begin{align}
\label{markov}Y_{1}- (X_{1},X_{3})-X_{2} \mbox{ and } Y_{2}- (X_{2}%
,X_{3})-X_{1}.
\end{align}
The following proposition, whose proof we present in \cite{Gunduz:ITtbs}, provides an outer bound for
the capacity region.
\begin{prop}
\label{p:outerbound} Assuming that the Markov chain conditions (\ref{markov})
hold for any channel input distribution, a rate triplet ($R_{1},R_{2},R_{3}$)
with $R_{j}\geq0$, $j=1,2,3,$ is achievable only if
\begin{align*}
R_{1}  &  \leq I(X_{1};Y_{3}|U_{1}, X_{2}, X_{3},Q),\\
R_{2}  &  \leq I(X_{2};Y_{3}|U_{2}, X_{1},X_{3},Q),\\
R_{3}  &  \leq\min\{I(X_{3};Y_{1}|X_{1},X_{2},U_{1},U_{2},Q),\\
& ~~~~ I(X_{3};Y_{2}|X_{1},X_{2},U_{1},U_{2},Q)\},\\
R_{1}+R_{3}  &  \leq\min\{I(X_{1},X_{3};Y_{1}|U_{2},Q),\\
& ~~~~ I(X_{3};Y_{2}|X_{2},U_{2},Q)\},\\
R_{2}+R_{3}  &  \leq\min\{I(X_{3};Y_{1}|X_{1},U_{1},Q),\\
& ~~~~ I(X_{2},X_{3};Y_{2}|U_{1},Q)\}\\
R_{1}+R_{2}+R_{3}  &  \leq\min\{I(X_{1},X_{3};Y_{1}|Q),I(X_{2},X_{3}%
;Y_{2}|Q)\}
\end{align*}
for some auxiliary random variables $U_{1}, U_{2}$ and $Q$ satisfying the
joint distribution $p(q)p(x_{1},u_{1}|q)p(x_{2},u_{2}|q)$ $p(x_{3}|u_{1}%
,u_{2},q) p(y_{1},y_{2},y_{3}|x_{1},x_{2},x_{3})$.
\end{prop}


By imposing the condition (\ref{markov}) on the DF achievable rate region of
Proposition \ref{p:achDF}, it can be easily seen that the only difference
between the outer bound and the achievable region with DF (\ref{ach:DF}) is
that the latter contains the additional constraint (\ref{ach 3}), which
generally reduces the rate region. The constraint (\ref{ach 3}) accounts for
the fact that the DF scheme prescribes both messages $W_{1}$ and $W_{2}$ to be decoded at the relay
terminal. The following remark provides two examples in which the DF scheme
achieves the outer bound in Proposition \ref{p:outerbound} and thus the
capacity region. In both cases, the multiple access interference at the relay
terminal is eliminated from the problem setup so that the condition
(\ref{ach 3}) does not limit the performance of DF.

\begin{rem}
In addition to the Markov conditions in (\ref{markov}), consider orthogonal
channels from the two users to the relay terminal, that is, we have
$Y_{3}\triangleq(Y_{31},Y_{32})$, where $Y_{3k}$ depends only on inputs
$X_{k}$ and $X_{3}$ for $k=1,2$; that is, we assume $X_{1}-(X_{2}%
,X_{3})-Y_{32}$ and $X_{2}-(X_{1},X_{3})-Y_{31}$ form Markov chains for any
input distribution. Then, it is easy to see that the sum-rate constraint at
the relay terminal is redundant and hence the outer bound in Proposition
\ref{p:outerbound} and the achievable rate region with DF in Proposition
\ref{p:achDF} match, yielding the full capacity region for this scenario. As
another example where DF\ is optimal, we consider a \textit{relay multicast
channel} setup, in which a single relay helps transmitter 1 to multicast its
message $W_{1}$ to both receivers, i.e., $R_{2}=R_{3}=0$ and $X_{2}=\emptyset
$. For such a setup, under the assumption that $X_{1}-X_{3}-Y_{2}$ forms a
Markov chain, the achievable rate with DF relaying in Proposition
\ref{p:achDF} and the above outer bound match. Specifically, the capacity $C$
for this \textit{multicast relay channel} is given by
\begin{equation}
C=\max_{p(x_{1},x_{3})}\min\{I(X_{1};Y_{3}|X_{3}),\text{ }I(X_{1},X_{3}%
;Y_{1}),\text{ }I(X_{3};Y_{2})\}.\nonumber
\end{equation}
\end{rem}

\begin{rem}
This work is limited to random coding techniques to provide achievable rate regions. However, the structure of the network under the Markov assumptions in (\ref{markov}) can be exploited using structured codes rather than the random codes. This enables the relay to decode only the modulo sum of the messages relaxing the sum-rate constraint at the relay. We explore this in more detail in \cite{Gunduz:ITtbs} and show that structured coding schemes achieve the capacity in certain scenarios.
\end{rem}

\section{Gaussian Compound MAC with a Relay}

\label{p:GcMACr}

A Gaussian cMACr\ satisfying the Markov conditions (\ref{markov}) is given by
\begin{subequations}
\label{G mac relay}%
\begin{align}
Y_{1}  &  =X_{1}+\eta X_{3}+Z_{1}\\
Y_{2}  &  =X_{2}+\eta X_{3}+Z_{2}\\
Y_{3}  &  =\gamma(X_{1}+X_{2})+Z_{3},
\end{align}
where $\gamma\geq0$ is the channel gain from the users to the relay and
$\eta\geq0$ is the channel gain from the relay to both receiver 1 and receiver
2. The noise components $Z_{i}$, $i=1,2,3$ are i.i.d. zero-mean unit variance
Gaussian random variables. We enforce the average power constraints $\frac
1{n}\sum\limits_{i=1}^{n}E[X_{ji}^{2}]\leq P_{j}$ for $j=1,2,3$. We define $C(x)=\frac12\log(1+x)$ for $x\in \mathds{R}^+$. Considering
for simplicity the case $R_{3}=0,$ we have the following result.
\end{subequations}
\begin{prop}
The following rate region is achievable for the Gaussian cMACr characterized
by (\ref{G mac relay}) with DF:
\begin{subequations}\label{region compound mac}%
\begin{align}
R_{1}  &  \leq\min\left\{
\begin{array}
[c]{l}C \left( \gamma^{2}P_{1}\left(  1-\frac{\alpha_{1}\alpha
_{3}^{\prime}}{1-\alpha_{2}\alpha_{3}^{\prime\prime}}\right)  \right)  ,\\
C\left(P_{1}+\eta^{2}P_{3}(1-\alpha_{3}^{\prime\prime})\right)
\end{array}
\right\} ,\\
R_{2}  &  \leq\min\left\{
\begin{array}
[c]{l}%
C\left(\gamma^{2}P_{2}\left(  1-\frac{\alpha_{2}\alpha
_{3}^{\prime\prime}}{1-\alpha_{1}\alpha_{3}^{\prime}}\right)  \right)  ,\\
C\left(P_{2}+\eta^{2}P_{3}(1-\alpha_{3}^{\prime})\right)
\end{array}
\right\}
\end{align}
and
\begin{align}
&  R_{1}+R_{2} \leq\min\left\{  C\left(P_{1}+\eta^{2}P_{3}+2\eta\sqrt{\alpha_{1}\alpha_{3}^{\prime}P_{1}P_{3}%
}\right) , \right. \nonumber\\
& ~~~~~~~~~~~~~~~~~~~~~ C\left(P_{2}+\eta^{2}P_{3}+2\eta\sqrt{\alpha_{1}\alpha
_{3}^{\prime\prime}P_{2}P_{3}} \right)  ,\nonumber\\
&  \left.  C\left(\gamma^{2}(P_{1}+P_{2})\left(  1-\frac
{(\sqrt{\alpha_{1}\alpha_{3}^{\prime}P_{1}}+\sqrt{\alpha_{2}\alpha_{3}%
^{\prime\prime}P_{2}})^{2}}{P_{1}+P_{2}}\right)  \right)  \right\}  ,
\label{sum rate gaussian}%
\end{align}
with the union taken over the parameters $0\leq\alpha_{1},\alpha_{2}%
,\alpha_{3}^{\prime},\alpha_{3}^{\prime\prime}\leq1$ and $\alpha_{3}^{\prime
}+\alpha_{3}^{\prime\prime}\leq1.$ An outer bound to the capacity region is
given by (\ref{region compound mac}) without the last sum-rate constraint in
(\ref{sum rate gaussian}).
\end{subequations}
\end{prop}

\begin{proof}
The proof can be found in \cite{Gunduz:ITtbs}
\end{proof}

It is noted that the parameters $\alpha_{3}^{\prime}$ and $\alpha_{3}%
^{\prime\prime}$ represent the fractions of power that the relay uses to
cooperate with transmitter 1 and 2, respectively.

Next, we characterize the achievable rate region for the Gaussian setup with
the CF strategy of Proposition \ref{p:achCF}. Here, we assume a Gaussian
quantization codebook without claiming optimality.

\begin{prop}
The following rate region is achievable for the Gaussian cMACr characterized
by (\ref{G mac relay}):
\begin{subequations}
\label{ach:CF}%
\begin{align}
R_{1}  &  \leq C\left(\frac{\gamma^{2}\alpha_{1}P_{1}}{1+N_{q}%
}\right) \\
R_{2}  &  \leq C\left(\frac{\gamma^{2}\alpha_{2}P_{2}}{1+N_{q}%
}\right) \mbox{ and }
\end{align}
\begin{align}
R_1 + R_2 & \leq C(\bar{P}) + C\left(\frac{\gamma^2(\alpha_1P_1+\alpha_2P_2)}{1+N_q}\right)
\end{align}
where
\end{subequations}
\[
N_{q}=\frac{1+\gamma^{2}(\alpha_{1}P_{1}\alpha_{2}P_{2}+\alpha_{1}P_{1}%
+\alpha_{2}P_{2})+\bar{P}}
{\eta^{2}P_{3}},
\]
for all $0\leq\alpha_{i}\leq1$, $i=1,2$, where we defined $\bar{P} = \min\{\alpha_{1}P_{1}, \alpha_2 P_2\}$.
\end{prop}

\begin{figure}[ptb]
\begin{center}
\includegraphics[width=3.5in]
{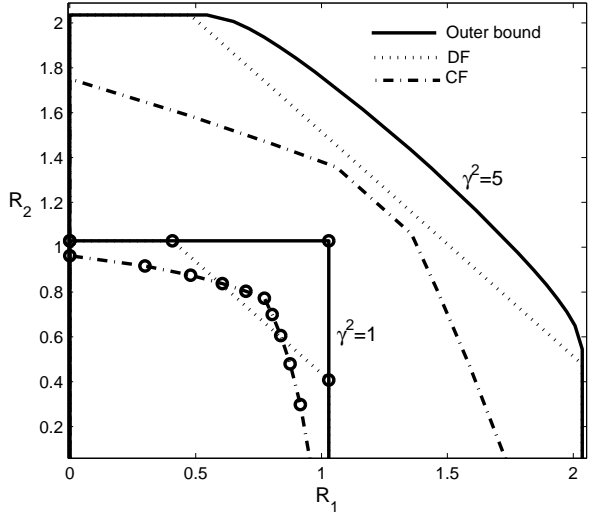}
\end{center}
\caption{Achievable rate region and outer bound for $P_{1}=P_{2}=P_{3}=5~dB$,
$\eta^{2}=10$ and different values of the channel gain from the terminals to
the relay, namely $\gamma^{2}=1,5.$}%
\label{compound region}%
\end{figure}

\subsubsection{Numerical examples}

Consider a cMACr with powers $P_{1}=P_{2}=P_{3}=5~dB$ and
channel gain $\eta^{2}=10$ from the relay to the two terminals. Fig.
\ref{compound region} shows the achievable rate region and the outer bound for
different values of the channel gain from the terminals to the relay, namely
$\gamma^{2}=1$ and $\gamma^2=5$. It can be seen that, if the channel to the relay is
weak, then CF improves upon DF at certain parts of the rate region. However, as $\gamma^{2}$ increases, DF gets very close to the outer bound dominating the CF rate region, since the sum rate constraint in DF scheme becomes less restricting.

In Fig. \ref{fig:sym_rate}, the symmetric rate achievable with DF and CF is compared
with the upper bound for $\gamma^{2}=1$ and $\eta^{2}=10$ versus
$P_{1}=P_{2}=P_{3}=P$. We see that while DF achieves a higher rate than CF and performs very close to the upper bound at low power values, with increasing power CF surpasses the DF rate for this channel setup. There is a constant bit gap between the upper bound and the achievable rates with DF and CF at higher power values.

\begin{figure}[ptb]
\begin{center}
\includegraphics[width=3.9in]
{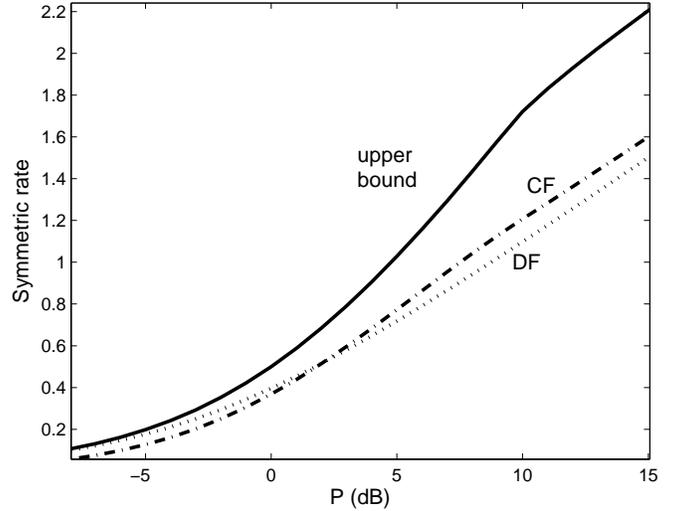}
\end{center}
\caption{Symmetric rate achievable with DF and CF strategies compared with the upper bound
for $\gamma^{2}=1$ and $\eta^{2}=10$ versus $P_{1}=P_{2}=P_{3}=P$. }%
\label{fig:sym_rate}%
\end{figure}

\section{Conclusions}

\label{s:conc}

We have considered a compound MAC with a relay terminal.
In this model, the relay simultaneously assists both transmitters
while multicasting its own information at the same time. We have first
characterized the capacity region for a MAC with a
cognitive relay and related models of partially cognitive relay and cognition
through finite capacity links. We then have used the coding technique that
achieves the capacity for these models to provide an achievable rate region
with DF relaying in the case of a general cMACr. We have also considered a CF
relaying scheme, in which the relay broadcasts a compressed version of
its received signal considering the received signals at the receivers as side
information. Here we have used a novel joint source-channel coding scheme to
improve the achievable rate region of the underlying multi-user channel coding problem. Strategies based on structured codes and physical layer network coding are studied in [9].


\begin{thebibliography}{99}                                                                                               \bibitem {Cover:IT:79} T. M. Cover and A. El Gamal, ``Capacity theorems for the
relay channel,'' \textit{IEEE Trans. Inform. Theory}, vol. 25, no. 5, pp.
572-584, Sep. 1979.

\bibitem {Kramer:IT:05} G. Kramer, M. Gastpar and P. Gupta, ``Cooperative
strategies and capacity theorems for relay networks,'' \textit{IEEE Trans.
Inform. Theory}, vol. 51, no. 9, pp. 3037--3063, Sep. 2005.

\bibitem {Willems:IT:83} F. M. J. Willems, ``The discrete memoryless multiple access
channel with partially cooperating encoders,'' \textit{IEEE Trans. Inform.
Theory}, vol. 29, no. 3, pp. 441--445, May 1983.



\bibitem {Gunduz:All:08} D. G\"{u}nd\"{u}z, E. Tuncel and J. Nayak, ``Achievable
rates for two-way relay channels,'' \textit{Proc. 46th Annual Allerton Conference on
Communication, Control and Computing}, Monticello, IL, USA, Sep. 2008.


\bibitem {Maric:MILCOM:07}I. Maric, A. Goldsmith and M. M\'{e}dard,
``Information-theoretic relaying for multicast in wireless networks,'' \textit{Proc.
IEEE Military Communications Conference (MILCOM)}, Orlando, FL, Oct. 2007.

\bibitem {Devroye:IT:06}N. Devroye, P. Mitran and V. Tarokh,
\textquotedblleft Achievable rates in cognitive radio
channels,\textquotedblright\ \textit{IEEE Trans. Inform. Theory}, vol. 52, no.
5, pp. 1813-1827, May 2006.


\bibitem {Han} T. S. Han, \textquotedblleft The capacity region of general multiple access channel with certain correlated
sources,\textquotedblright\ \textit{Inform. and Control}, vol. 40, no. 1, pp. 37-60, Jan. 1979.

\bibitem {Nayak:IT:08}J. Nayak, E. Tuncel and D. G\"{u}nd\"{u}z,
\textquotedblleft Wyner-Ziv coding over broadcast channels: Digital
schemes,\textquotedblright\ under revision, \textit{IEEE Trans. Inform. Theory}.

\bibitem {Gunduz:ITtbs} D. G\"{u}nd\"{u}z, O. Simeone, A. J. Goldsmith, H. V. Poor and S. Shamai, ``Multiple multicasts with the help of a relay,'' submitted, \textit{IEEE Trans. Inform. Theory}.

\end{thebibliography}
\end{document}